\documentclass[10pt, journal]{IEEEtran}

\usepackage{xspace,amsmath,amssymb,amsfonts,epsfig,syntonly}
\usepackage{cite,bm,color,url,textcomp}
\usepackage{algorithmicx,algorithm,algpseudocode}
\usepackage{epstopdf}
\usepackage{empheq}
\usepackage{balance}

\newtheorem{lemma}{\hspace{-11pt}\bf Lemma}

\newtheorem{proposition}{\hspace{-11pt}\bf Proposition}
\newtheorem{theorem}{\hspace{-11pt}\bf Theorem}

\newtheorem{remark}{\hspace{-11pt}\bf Remark}

\long\def\symbolfootnote[#1]#2{\begingroup
\def\thefootnote{\fnsymbol{footnote}}
\footnote[#1]{#2}\endgroup}

\psfull

\allowdisplaybreaks[2]

\begin{document}
\title{Two-Scale Stochastic Control for Integrated Multipoint Communication Systems with Renewables}
\author{Xin Wang,~\IEEEmembership{Senior Member, IEEE}, Xiaojing Chen,~\IEEEmembership{Student Member, IEEE},
Tianyi Chen,~\IEEEmembership{Student Member, IEEE}, Longbo Huang,~\IEEEmembership{Member, IEEE}, and Georgios B. Giannakis,~\IEEEmembership{Fellow, IEEE}\\
\thanks {Work in this paper was supported by the Program for New Century Excellent Talents in University, the Innovation Program of Shanghai Municipal Education Commission; US NSF grants ECCS-1509005, ECCS-1508993, CCF-1423316, CCF-1442686, and ECCS-1202135. The work of L. Huang was supported in part by the National Basic Research Program of China Grant 2011CBA00300, 2011CBA00301, the National Natural Science Foundation of China Grant 61033001, 61361136003, 61303195, and the China youth 1000-talent grant.}

\thanks{
X. Wang and X. Chen are with the Key Laboratory for Information Science of Electromagnetic Waves (MoE), Department of Communication Science and Engineering, Fudan University, 220 Han Dan Road, Shanghai, China. Emails:~\{xwang11, 13210720095\}@fudan.edu.cn. X.~Wang is also with the Department of Computer and Electrical Engineering and Computer Science, Florida Atlantic University, Boca Raton, FL 33431 USA.

T. Chen and G. B. Giannakis are with the Department of Electrical and Computer Engineering and the Digital Technology Center, University of Minnesota, Minneapolis, MN 55455 USA. Emails:~\{chen3827, georgios\}@umn.edu.

L. Huang is with the Institute for Interdisciplinary Information Sciences, Tsinghua University, Beijing, China. Email: longbohuang@tsinghua.edu.cn.

X. Wang, X. Chen, and T. Chen contributed equally to this work.}
}

\markboth{IEEE TRANSACTION ON SMART GRID (to appear)}{}

\maketitle

\begin{abstract}
Increasing threats of global warming and climate changes call for an energy-efficient and sustainable design of future wireless communication systems. To this end, a novel two-scale stochastic control framework is put forth for smart-grid powered coordinated multi-point (CoMP) systems. Taking into account renewable energy sources (RES), dynamic pricing, two-way energy trading facilities and imperfect energy storage devices, the energy management task is formulated as an infinite-horizon optimization problem minimizing the time-averaged energy transaction cost, subject to the users' quality of service (QoS) requirements. Leveraging the Lyapunov optimization approach as well as the stochastic subgradient method, a two-scale online control (TS-OC) approach is developed for the resultant smart-grid powered CoMP systems. Using only historical data, the proposed TS-OC makes online control decisions at two timescales, and features a provably feasible and asymptotically near-optimal solution. Numerical tests further corroborate the theoretical analysis, and demonstrate the merits of the proposed approach.
\end{abstract}


\section{Introduction}

Interference is a major obstacle in wireless communication systems due to their broadcast nature, and becomes more severe in next-generation spectrum- and energy-constrained cellular networks with smaller cells and more flexible frequency reuse~\cite{Hwa13}. With ever increasing demand for energy-efficient transmissions, coordinated multi-point processing (CoMP) has been proposed as a promising paradigm for efficient inter-cell interference management in heterogeneous networks (HetNets)~\cite{Irm11}. In CoMP systems, base stations (BSs) are partitioned into clusters, where BSs per cluster perform coordinated beamforming to serve the users~\cite{Dal10, Zha09, Ng10}.
As the number of BSs in HetNets increases, their electricity consumption constitutes a major part of their operational expenditure, and contributes a considerable portion to the global \emph{carbon footprint}~\cite{Oh11}.
Fortunately, emerging characteristics of smart grids offer ample opportunities to achieve both energy-efficient and environmentally-friendly communication solutions. Such characteristics include high penetration of renewable energy sources (RES), two-way energy trading, and dynamic pricing based demand-side management (DSM)~\cite{Zha13, Liu10, Gian13}.
In this context, energy-efficient ``green'' communication solutions have been proposed for their economic and ecological merits \cite{Dal10, Zha09, Oh11, Ng10}. Driven by the need of sustainable ``green communications,'' manufacturers and network operators such as Ericsson, Huawei, Vodafone and China Mobile have started developing ``green'' BSs that can be jointly supplied by the persistent power sources from the main electric grid as well as from harvested renewable energy sources (e.g., solar and wind) \cite{Huawei, Li15}. It is expected that renewable powered BSs will be widely deployed to support future-generation cellular systems.

A few recent works have considered the smart-grid powered CoMP transmissions~\cite{Xu13, Xu15, WangTWC15, WangJSAC15}.
Assuming that the energy harvested from RES is accurately available \emph{a priori} through e.g., forecasting, \cite{Xu13} and~\cite{Xu15} considered the energy-efficient resource allocation for RES-powered CoMP downlinks.
Building on realistic models, our recent work dealt with robust energy management and transmit-beamforming designs that minimize the worst-case energy transaction cost for the CoMP downlink with RES and DSM \cite{WangTWC15}. Leveraging novel stochastic optimization tools \cite{Urg11, Lak14, Sun14}, we further developed an efficient approach to obtain a feasible and asymptotically optimal online control scheme for smart-grid powered CoMP systems, without knowing the distributions of involved random variables \cite{WangJSAC15}.

A salient assumption in \cite{Xu13, Xu15, WangTWC15, WangJSAC15} is that all involved resource allocation tasks are performed in a single time scale. However, RES and wireless channel dynamics typically evolve over different time scales in practice. 
Development of two-scale control schemes is then well motivated for CoMP systems with RES. In related contexts, a few stochastic optimization based two-scale control schemes were recently proposed and analyzed in~\cite{Yao12, Deng13, Prasad14, Yu16}. Extending the traditional Lyapunov optimization approach \cite{Urg11, Lak14, Sun14}, \cite{Yao12} introduced a two-scale control algorithm that makes distributed routing and server management decisions to reduce power cost for large-scale data centers. Based on a similar approach, \cite{Deng13} developed a so-called MultiGreen algorithm for data centers, which allows cloud service providers to make energy transactions at two time scales for minimum operational cost. As far as wireless communications are concerned, \cite{Prasad14} performed joint precoder assignment, user association, and channel resource scheduling for HetNets with non-ideal backhaul; while \cite{Yu16} introduced a two-timescale approach for network selection and subchannel allocation for integrated cellular and Wi-Fi networks with an emphasis on using predictive future information. Note that however, neither \cite{Prasad14} nor \cite{Yu16} considers the diversity of energy prices in fast/slow-timescale energy markets, and the energy leakage effects in the energy management task.

\section{}
In the present paper, we develop a two-scale online control (TS-OC) approach for smart-grid powered CoMP systems considering RES, dynamic pricing, two-way energy trading facilities and imperfect energy storage devices. Suppose that the RES harvesting occurs at the BSs over a slow timescale relative to the coherence time of wireless channels. The proposed scheme performs an ahead-of-time (e.g., 15-minute ahead, or, hour-ahead) energy planning upon RES arrivals, while deciding real-time energy balancing and transmit-beamforming schedules per channel coherence time slot. Specifically, the TS-OC determines the amount of energy to trade (purchase or sell) with the ahead-of-time wholesale market based on RES generation, as the basic energy supply for all the time slots within a RES harvesting interval. On the other hand, it decides the amount of energy to trade with the real-time market, energy charging to (or discharging from) the batteries, as well as the coordinated transmit-beamformers to guarantee the users' quality of service (QoS) per time slot.
Generalizing the Lyapunov optimization techniques in \cite{Yao12,Deng13,Prasad14, Yu16,Qin15}, we propose a synergetic framework to design and analyze such a two-scale dynamic management scheme to minimize the long-term time-averaged energy transaction cost of the CoMP transmissions, without knowing the distributions of the random channel, RES, and energy price processes. The main contributions of our work are summarized as follows.
\begin{itemize}
\item Leveraging the ahead-of-time and real-time electricity markets, and building on our generalized system models in \cite{WangTWC15, WangJSAC15}, a novel two-scale optimization framework is developed to facilitate the dynamic resource management for smart-grid powered CoMP systems with RES and channel dynamics at different time scales.

\item While \cite{WangJSAC15,Yao12} and \cite{Deng13} do not account for battery degeneration (energy leakage), we integrate the modified Lyapunov optimization technique into the two-scale stochastic optimization approach to leverage the diversity of energy prices along with the energy leakage effects on the dynamic energy management task.

\item Using only past channel and energy-price realizations, a novel stochastic subgradient approach is developed to solve the ahead-of-time energy planning (sub-)problem, which is suitable for a general family of continuous distributions, and avoids constructing the histogram estimate which is computationally cumbersome, especially for high-dimensional vector of random optimization variables.

\item Rigorous analysis is presented to justify the feasibility and quantify the optimality gap for the proposed two-scale online control algorithm.
\end{itemize}

The rest of the paper is organized as follows. The system models are described in Section II. The proposed dynamic resource management scheme
is developed in Section III. Performance analysis is the subject of Section IV. Numerical tests are provided in Section V,
followed by concluding remarks in Section VI.

\noindent \emph{Notation}. Boldface lower (upper) case letters represent vectors (matrices); $\mathbb{C}^{N}$ and $\mathbb{R}^{N \times M}$ stand for spaces of $N \times 1$ complex vectors and $N \times M$ real matrices, respectively; $(\cdot)'$ denotes transpose, and $(\cdot)^H$ conjugate transpose; $\text{diag}(a_1, \ldots, a_M)$ denotes a diagonal matrix with diagonal elements $a_1, \ldots, a_M$;
$| \cdot |$ the magnitude of a complex scalar; and $\mathbb{E}$ denotes expectation.

\section{System Models}


Consider a cluster-based CoMP downlink setup, where a set ${\cal I}:=\{1,\ldots, I\}$ of distributed BSs (e.g., macro/micro/pico BSs) is selected to serve a set ${\cal K}:=\{1,\ldots, K\}$ of mobile users, as in e.g., \cite{WangTWC15, WangJSAC15}. Each BS is equipped with $M\geq 1$ transmit antennas, whereas each user has a single receive antenna. Suppose that through the smart-grid infrastructure conventional power generation is available, but each BS can also harvest RES (through e.g., solar panels and/or wind turbines), and it has an energy storage device (i.e., battery) to save the harvested energy. Relying on a two-way energy trading facility, the BS can also buy energy from or sell energy to the main grid at dynamically changing market prices. For the CoMP cluster, there is a low-latency backhaul network connecting the set of BSs to a central controller~\cite{Zha09}, which coordinates energy trading as well as cooperative communication. This central entity can collect both communication data (transmit messages, channel state information) from each BS through the cellular backhaul links, as well as the energy information (energy purchase/selling prices, energy queue sizes) via smart meters installed at BSs, and the grid-deployed communication/control links connecting them.\footnote{Perfect channel state information will be assumed hereafter, but the proposed formulation can readily account for the channel estimation errors to robustify the beamforming design; see e.g., \cite{WangTWC15, WangJSAC15}. In addition, generalizations are possible to incorporate imperfect energy queue information based on the Lyapunov optimization framework in \cite{Deng13}. Although their detailed study falls out of the present paper's scope, such imperfections are not expected to substantially affect the effectiveness of the proposed scheme. }

\begin{figure}\label{sec:model}
\centering
\includegraphics[width=0.5\textwidth]{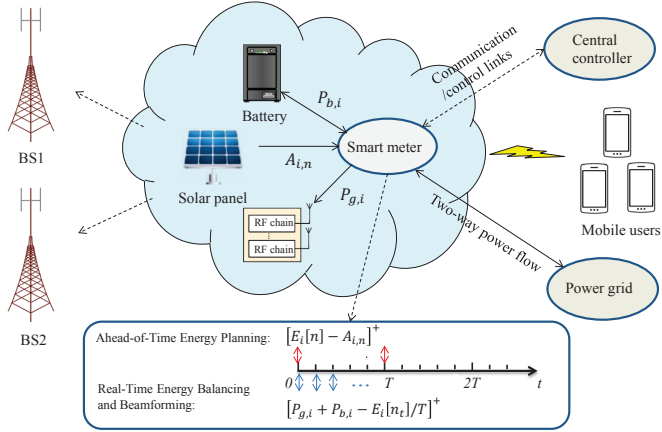}
\vspace{-0.6cm}
\caption{A smart grid powered CoMP system.
Two BSs with local renewable energy harvesting and storage devices implement two-way energy trading with the main grid.
}
\vspace{-0.6cm}
\label{fig:system}
\end{figure}


As the RES and wireless channel dynamics emerge typically at different time scales in practice, we propose a two-scale control mechanism. As shown in Fig.~\ref{fig:system}, time is divided in slots of length smaller than the coherence time of the wireless channels. For convenience, the slot duration is normalized to unity; thus, the terms ``energy'' and ``power'' can be used interchangeably. On the other hand, we define the (virtual) ``coarse-grained'' time intervals in accordance with the slow RES harvesting scale, with each coarse-grained interval consisting of $T$ time slots.

\subsection{Ahead-of-Time Energy Planning}

At the beginning of each ``coarse-grained'' interval, namely at time $t=nT$, $n=1, 2, \ldots$, let $A_{i,n}$ denote the RES amount collected per BS $i \in {\cal I}$, and $\boldsymbol{A}_n:=[A_{1,n},\ldots,A_{I,n}]'$. With $\boldsymbol{A}_n$ available, an energy planner at the central unity decides the energy amounts $E_i[n]$, $\forall i$, to be used in the next $T$ slots per BS $i$. With a two-way energy trading facility, the BSs then either purchase energy from the main grid according to their shortage, or sell their surplus energy to the grid at a fair price in order to reduce operational costs. Specifically, following the decision, BS $i$ contributes its RES amount $A_{i,n}$ to the main grid and requests the grid to supply an average energy amount of $E_i[n]/T$ per slot $t=nT, \ldots, (n+1)T-1$.

RES is assumed harvested for free after deployment. Given the requested energy $E_i[n]$ and the harvested energy $A_{i,n}$, the shortage energy that is purchased from the grid for BS $i$ is clearly $[E_i[n]-A_{i,n}]^{+}$; or, the surplus energy that is sold to the grid is $[A_{i,n}-E_i[n]]^{+}$, where $[a]^+:= \max\{a, 0\}$.
Depending on the difference $(E_i[n]-A_{i,n})$, the BS $i$ either buys electricity from the grid with the ahead-of-time (i.e., long-term) price $\alpha_n^{{\rm lt}}$, or sells electricity to the grid with price $\beta_n^{{\rm lt}}$ for profit (the latter leads to a negative cost).
Notwithstanding, we shall always set $\alpha_n^{{\rm lt}} > \beta_n^{{\rm lt}}$  to avoid meaningless buy-and-sell activities of the BSs for profit. The transaction cost with BS~$i$ for such an energy planning is therefore given by
\begin{equation}
\label{eq.lt-tranct}
G^{{\rm lt}}(E_i[n]):= \alpha_n^{{\rm lt}}[E_i[n]\!-\!A_{i,n}]^{+} \!-\! \beta_n^{{\rm lt}} [A_{i,n}\!-\!E_i[n]]^{+}.
\end{equation}
For conciseness, we concatenate into a single random vector all the random variables evolving at this slow timescale; i.e., $\boldsymbol{\xi}_n^{\rm lt}:=\{\alpha_n^{\rm lt}, \beta_n^{\rm lt},\mathbf{A}_n, \forall n\}$.

\subsection{CoMP Downlink Transmissions}

Per slot $t$, let $\mathbf{h}_{ik,t} \in \mathbb{C}^{M}$ denote the vector channel from BS $i$ to user $k$, $\forall i \in {\cal I}$, $\forall k \in \mathcal{K}$; let $\mathbf{h}_{k,t} := [\mathbf{h}_{1k,t}', \ldots, \mathbf{h}_{Ik,t}']'$ collect the channel vectors from all BSs to user $k$, and $\mathbf{H}_t:=[\mathbf{h}_{1,t}. \ldots, \mathbf{h}_{K,t}]$. With linear transmit beamforming performed across BSs, the vector signal transmitted  to user $k$ is: $\mathbf{q}_k(t) =\mathbf{w}_k(t) s_k(t)$, $\forall k$, where $s_k(t)$ denotes the information-bearing scalar symbol with unit-energy, and $\mathbf{w}_k(t) \in \mathbb{C}^{MI}$ denotes the beamforming vector across the BSs serving user $k$. The received vector at slot $t$ for user $k$ is therefore
\begin{align}\label{eq.yk}
y_k(t) =\mathbf{h}_{k,t}^H \mathbf{q}_k(t) + \sum_{l\neq k} \mathbf{h}_{k,t}^H \mathbf{q}_l(t) + n_k(t)
\end{align}
where $\mathbf{h}_{k,t}^H \mathbf{q}_k(t)$ is the desired signal of user $k$, $\sum_{l\neq k} \mathbf{h}_{k,t}^H \mathbf{q}_l(t)$ is the inter-user interference from the same cluster, and $n_k(t)$ denotes additive noise, which may also include the downlink interference from other BSs outside user $k$'s cluster. It is assumed that $n_k(t)$ is a circularly symmetric complex Gaussian (CSCG) random variable with zero mean and variance $\sigma_k^2$.

The signal-to-interference-plus-noise ratio (SINR) at user $k$ can be expressed as
\begin{equation}\label{sinr}
    \text{SINR}_k(\{\mathbf{w}_k(t)\})= \frac{|\mathbf{h}_{k,t}^H \mathbf{w}_k(t)|^2}{\sum_{l\neq k} (|\mathbf{h}_{k,t}^H \mathbf{w}_l(t)|^2) + \sigma_k^2}~.
\end{equation}
The transmit power at each BS $i$ clearly is given by
\begin{equation}
    P_{x,i}(t) = \sum_{k \in \mathcal{K}} {\mathbf{w}_k^H(t)} \mathbf{B}_i \mathbf{w}_k(t)
\end{equation}
where the matrix $$\mathbf{B}_i:=\text{diag}\left(\underbrace{0, \ldots, 0}_{(i-1)M}, \underbrace{1, \ldots, 1}_{M}, \underbrace{0, \ldots, 0}_{(I-i)M}\right) \in \mathbb{R}^{MI\times MI}$$ selects the corresponding rows out of $\{\mathbf{w}_k(t)\}_{k\in \mathcal{K}}$ to form the $i$-th BS's transmit-beamforming vector of size $M\times 1$.

To guarantee QoS per slot user $k$, it is required that the central controller selects a set of $\{\mathbf{w}_k(t)\}$ satisfying [cf. \eqref{sinr}]
\begin{equation}\label{eq.sinr}
    \text{SINR}_k (\{\mathbf{w}_k(t)\}) \geq \gamma_k, \quad \forall k
\end{equation}
where $\gamma_k$ denotes the target SINR value per user $k$.

\subsection{Real-Time Energy Balancing}

For the $i$-th BS, the total energy consumption $P_{g,i}(t)$ per slot $t$ includes the transmission-related power $P_{x,i}(t)$, and the rest that is due to other components such as air conditioning, data processor, and circuits, which can be generally modeled as a constant power, $P_c > 0$ \cite{Xu15}. We further suppose that $P_{g,i}(t)$ is bounded by $P_g^{\max}$. Namely,
\begin{equation}\label{eq.pg}
    P_{g,i}(t)=P_c+\sum_{k \in \mathcal{K}} {\mathbf{w}_k^H(t)} \mathbf{B}_i \mathbf{w}_k(t) \leq P_g^{\max}, \quad \forall i.
\end{equation}

\begin{figure}
\centering
\includegraphics[width=0.5\textwidth]{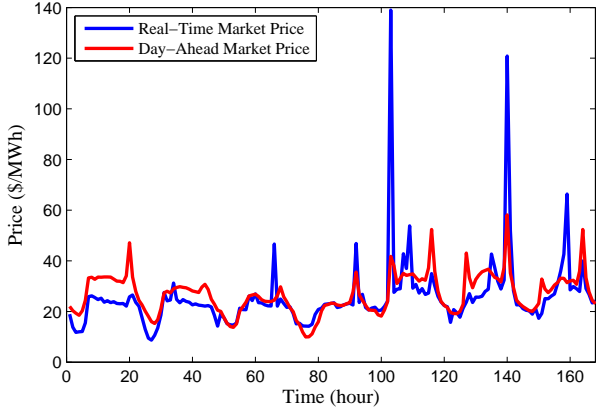}
\caption{Hourly price trend for day-ahead and real-time electricity markets during Oct. 01-07, 2015 \cite{PJM}.}
\label{fig: Price}
\vspace{-0.3cm}
\end{figure}

Per slot $t$, the energy supply available from the ahead-of-time planning may not exactly meet the actual demand at BS~$i$. Hence, the BS $i$ is also allowed to perform real-time energy trading with the main grid to balance its supply with demand. Let $P_i(t)$ denote the real-time energy amount that is purchased from ($P_i(t) >0$) or sold to ($P_i(t) <0$) the grid by BS $i$. Let $\alpha_t^{\rm rt}$ and $\beta_t^{{\rm rt}}$ ($\alpha_t^{\rm rt} > \beta_t^{{\rm rt}}$) denote the real-time energy purchase and selling prices, respectively. Then the real-time energy transaction cost for BS $i$ is
\begin{equation}
\label{eq.rt-tranct}
G^{{\rm rt}}(P_i(t)):= \alpha_t^{\rm rt}[P_i(t)]^{+} - \beta_t^{{\rm rt}} [-P_i(t)]^{+}.
\end{equation}
Fig.~\ref{fig: Price} depicts the day-ahead and real-time energy prices in the Pennsylvania-Jersey-Maryland (PJM) wholesale market~\cite{PJM}. In practice, the average purchase price in the real-time market tends to be no lower than that in the day-ahead market; that is, $\mathbb{E}\{\alpha_t^{\rm rt}\} \geq \mathbb{E}\{\alpha_n^{{\rm lt}}\}$; similarly, we have $\mathbb{E}\{\beta_t^{{\rm rt}}\} \leq \mathbb{E}\{\beta_n^{{\rm lt}}\}$. Again, we use a random vector $\boldsymbol{\xi}_t^{\rm rt}:=\{\alpha_t^{\rm rt}, \beta_t^{{\rm rt}},\mathbf{H}_t, \forall t\}$ to collect all random variables evolving at the fast timescale.

\subsection{Energy Storage with Degeneration}

As energy consumption will become a major concern of the future large-scale cellular networks, uninterrupted power supply type storage units can be installed at the BSs to prevent power outages, and provide opportunities to optimize the BSs' electricity bills. 
Different from the ideal battery models in \cite{Xu13, Xu15, WangTWC15, WangJSAC15, Deng13}, we consider here a practical battery with degeneration (i.e., energy leakage over time even in the absence of discharging) as in \cite{Qin15}.

For the battery of the $i$-th BS, let $C_i(0)$ denote the initial amount of stored energy, and $C_i(t)$ its state of charge (SoC) at the beginning of time slot $t$. The battery is assumed to have a finite capacity $C^{\max}$. Furthermore, for reliability purposes, it might be required to ensure that a minimum energy level $C^{\min}$ is maintained at all times. Let $P_{b,i}(t)$ denote the energy delivered to or drawn from the battery at slot $t$, which amounts to either charging ($P_{b,i}(t)>0$) or discharging ($P_{b,i}(t)<0$). The stored energy then obeys the dynamic equation
\begin{equation}\label{eq.Ci}
    C_i(t+1) = \eta C_i(t)+P_{b,i}(t), ~ C^{\min} \leq C_i(t) \leq C^{\max}, ~ \forall i
\end{equation}
where $\eta \in (0,1]$ denotes the storage efficiency (e.g., $\eta =0.9$ means that 10\% of the stored energy will be ``leaked'' over a slot, even in the absence of discharging).

The amount of power (dis)charged is also assumed bounded by
\begin{equation}\label{eq.pb}
     P_{b}^{\min} \leq P_{b,i}(t) \leq P_{b}^{\max}, \quad \forall i
\end{equation}
where $P_{b}^{\min} <0$ and $P_{b}^{\max}>0$ are introduced by physical constraints.

With $n_t := \lfloor \frac{t}{T} \rfloor$ and consideration of $P_{b,i}(t)$, we have the following demand-and-supply balance equation per slot $t$:
\begin{equation}\label{eq.balance}
   P_c+\sum_{k \in \mathcal{K}} {\mathbf{w}_k^H(t)} \mathbf{B}_i \mathbf{w}_k(t) + P_{b,i}(t) = \frac{E_i[n_t]}{T} + P_i(t), \;\forall i.
\end{equation}


\section{Dynamic Resource Management Scheme}

Note that the harvested RES amounts $\{\boldsymbol{A}_n, \forall n\}$, the ahead-of-time prices $\{\alpha_n^{{\rm lt}}, \beta_n^{{\rm lt}}, \forall n\}$, the real-time prices $\{\alpha_t^{\rm rt}, \beta_t^{{\rm rt}}, \forall t\}$, and the wireless channel matrices $\{\mathbf{H}_t, \forall t\}$ are all random. The smart-grid powered CoMP downlink to be controlled is a stochastic system. The goal is to design an online resource management scheme that chooses the ahead-of-time energy-trading amounts $\{E_i[n], \forall i\}$ at every $t=nT$, as well as the real-time energy-trading amounts $\{P_i(t), \forall i\}$, battery (dis)charging amounts $\{P_{b,i}(t), \forall i\}$, and the CoMP beamforming vectors $\{\mathbf{w}_k(t),\forall k\}$ per slot $t$, so as to minimize the expected total energy transaction cost, without knowing the distributions of the aforementioned random processes.

According to (\ref{eq.lt-tranct}) and (\ref{eq.rt-tranct}), define the energy transaction cost for BS $i$ per slot $t$ as:
\begin{equation}\label{eq.gt}
    \Phi_i(t):=\frac{1}{T} G^{{\rm lt}}(E_i[n_t]) + G^{{\rm rt}}(P_i(t)).
\end{equation}
Let ${\cal X}:=\{E_i[n], \forall i,n; P_i(t), P_{b,i}(t), C_i(t), \forall i, t; \mathbf{w}_k(t), \forall k, t\}$. The problem of interest is to find
\begin{equation}\label{eq.prob}
\begin{split}
   \Phi^{opt} := & \min_{{\cal X}} \;\lim_{N\rightarrow \infty} \frac{1}{NT} \sum_{t=0}^{NT-1} \sum_{i \in {\cal I}} \mathbb{E} \{\Phi_i(t)\} \\
   & \text{subject to} ~~~ (\ref{eq.sinr}), (\ref{eq.pg}), (\ref{eq.Ci}), (\ref{eq.pb}), (\ref{eq.balance}), ~~ \forall t
\end{split}
\end{equation}
where the expectations of $\Phi_i(t)$ are taken over all sources of randomness.
Note that here the constraints (\ref{eq.sinr}), (\ref{eq.pg}), (\ref{eq.Ci}), (\ref{eq.pb}), and (\ref{eq.balance}) are implicitly required to hold for every realization of the underlying random states $\boldsymbol{\xi}_t^{\rm rt}$ and $\boldsymbol{\xi}_n^{\rm lt}$.

\subsection{Two-Scale Online Control Algorithm}

(\ref{eq.prob}) is a stochastic optimization task. We next generalize and integrate the Lyapunov optimization techniques in \cite{Yao12,Deng13, Prasad14,Yu16,Qin15} to develop a TS-OC algorithm, which will be proven feasible, and asymptotically near-optimal for (\ref{eq.prob}). To start, assume the following two relatively mild conditions for the system parameters:
\begin{align}
    & P_b^{\max} \geq (1-\eta) C^{\min}  \label{eq.A1} \\
    & C^{\max}- C^{\min} \geq \frac{1-\eta^T}{1-\eta} (P_b^{\max} - P_b^{\min}). \label{eq.A2}
\end{align}

Condition (\ref{eq.A1}) simply implies that the energy leakage of the battery can be compensated by the charging. Condition (\ref{eq.A2}) requires that the allowable SoC range is large enough to accommodate the largest possible charging/discharging over $T$ time slots of each coarse-grained interval. This then makes the system ``controllable'' by our two-scale mechanism.

Our algorithm depends on two parameters, namely a ``queue perturbation'' parameter $\Gamma$, and a weight parameter $V$. Define $\bar{\alpha}:=\max\{\alpha_t^{\rm rt}, \forall t\}$ and $\underline{\beta}:=\min\{\beta_t^{{\rm rt}}, \forall t\}$. Derived from the feasibility requirement of the proposed algorithm (see the proof of Proposition 1 in the sequel), any pair $(\Gamma, V)$ that satisfies the following conditions can be used:
\begin{equation}\label{eq.GV}
    \Gamma^{\min} \leq \Gamma \leq \Gamma^{\max}, \quad 0 < V \leq V^{\max}
\end{equation}
where
\begin{flalign}
    & \Gamma^{\min} := \max_{\tau=1, \ldots, T} \left\{\frac{1}{\eta^{\tau}} (\frac{1-\eta^{\tau}}{1-\eta}P_b^{\max} - C^{\max})-V\underline{\beta}\right\}  \label{eq.GV1}\\
    & \Gamma^{\max} := \min_{\tau=1, \ldots, T}\left\{\frac{1}{\eta^{\tau}} (\frac{1-\eta^{\tau}}{1-\eta}P_b^{\min} - C^{\min})-V\bar{\alpha}\right\} \label{eq.GV2}\\
    & V^{\max}\!:=\!\! \min_{\tau=1, \ldots, T}\!\!\left\{\!\frac{C^{\max}\!-\!C^{\min} \!-\! \frac{1-\eta^{\tau}}{1-\eta}\!(P_b^{\max}\!-\!P_b^{\min})}{\eta^{\tau} (\bar{\alpha}-\underline{\beta})}\!\right\}\!.\!\label{eq.GV3}
\end{flalign}
Note that the interval for $V$ in (\ref{eq.GV}) is well defined under condition (\ref{eq.A2}), and the interval for $\Gamma$ is valid when $V\leq V^{\max}$.

We now present the proposed TS-OC algorithm:

\begin{itemize}

  \item \textbf{Initialization}: Select $\Gamma$ and $V$, and introduce a virtual queue $Q_i(0) := C_i(0) +\Gamma$, $\forall i$.

  \item \textbf{Ahead-of-time energy planning}: Per interval $\tau=nT$, observe a realization $\boldsymbol{\xi}_n^{\rm lt}$, and determine the energy amounts $\{E_i^*[n], \forall i\}$ by solving
      \begin{align}\label{eq.lt-prob}
       \min_{\{E_i^*[n]\}} \; \sum_{i\in \cal I} &\left\{V\Bigl[G^{{\rm lt}}(E_i[n])+\sum_{t = \tau}^{\tau+T-1} \mathbb{E}\{G^{{\rm rt}}(P_i(t))\}\Bigr] \right. \nonumber\\
      & ~~~\left. +\sum_{t = \tau}^{\tau+T-1} Q_i(\tau)\mathbb{E}\{P_{b,i}(t)\}\right\}\nonumber\\
       \text{s. t.} ~(\ref{eq.sinr}),(\ref{eq.pg})&,(\ref{eq.pb}), (\ref{eq.balance}), \quad\forall t =\tau, \ldots, \tau+T-1\!
      \end{align}
    where expectations are taken over $\boldsymbol{\xi}_t^{\rm rt}$.
    Then the BSs trade energy with the main grid based on $\{E_i^*[n], \forall i\}$, and request the grid to supply an average amount $E_i^*[n]/T$ per slot $t = \tau, \ldots, \tau+T-1$.

  \item \textbf{Energy balancing and beamforming schedule}: At every slot $t \in [nT, (n+1)T-1]$, observe a realization $\boldsymbol{\xi}_t^{\rm rt}$, and decide $\{P_i^*(t), P_{b,i}^*(t), \forall i; \mathbf{w}_k^*(t), \forall k\}$ by solving the following problem given $E_i[n]=E_i^*[n]$
      \begin{align}\label{eq.rt-prob}
        \min_{\{P_i^*(t), P_{b,i}^*(t), \mathbf{w}_k^*(t)\}} \; & \sum_{i\in \cal I} \left\{V G^{{\rm rt}}(P_i(t)) + Q_i(nT) P_{b,i}(t)\right\}\nonumber\\
        \text{s. t.}\; &~~ (\ref{eq.sinr}), (\ref{eq.pg}), (\ref{eq.pb}), (\ref{eq.balance}).
      \end{align}
      The BSs perform real-time energy trading with the main grid based on $\{P_i^*(t), \forall i\}$, and coordinated beamforming based on $\{\mathbf{w}_k^*(t), \forall k\}$.

  \item \textbf{Queue updates}: Per slot $t$, charge (or discharge) the battery based on $\{P_{b,i}^*(t)\}$, so that the stored energy $C_i(t+1) = \eta C_i(t) + P_{b,i}^*(t)$, $\forall i$; and update the virtual queues $Q_i(t) := C_i(t) +\Gamma, \forall i$.
\end{itemize}

\begin{remark}
Note that we use queue sizes $Q_i(\tau)$ instead of $Q_i(t)$ in problems (\ref{eq.lt-prob}) and (\ref{eq.rt-prob}); see also \cite{Yao12, Deng13}. Recall that the main design principle in Lyapunov optimization is to choose control actions that minimize $\sum_t \sum_i \mathbb{E}\left[V \Phi_i(t) + Q_i(t)P_{b,i}(t)\right]$. For the ahead-of-time energy planning, this requires a-priori knowledge of the future queue backlogs $Q_i(t)$ over slots $[\tau+1, \ldots, \tau+T-1]$ at time $\tau=nT$. It is impractical to assume that this information is available. For this reason, we simply approximate future queue backlog values as the current value at $\tau=nT$, i.e., $Q_i(t) \approx Q_i(\tau)$, $\forall t=\tau+1, \ldots, \tau+T-1$, in (\ref{eq.lt-prob}). To ensure that the real-time energy balancing and beamforming schedule solves the same problem as the ahead-of-time energy planning, we also use $Q_i(nT)$ in (\ref{eq.rt-prob}) although the real-time battery state of charge $Q_i(t)$ is available at slot $t$. Rigorous analysis shows that the performance penalty incurred by such an approximation does not affect the asymptotic optimality of the proposed stochastic control scheme. On the other hand, using $Q_i(t)$ in real-time energy balancing can be also suggested in practice. While our feasibility analysis affords such a modification, deriving the optimality gap is left for future research.
\end{remark}

Next, we develop efficient solvers of (\ref{eq.lt-prob}) and (\ref{eq.rt-prob}) to obtain the TS-OC algorithm.

\subsection{Real-Time Energy Balancing and Beamforming}

It is easy to argue that the objective (\ref{eq.rt-prob}) is convex. Indeed, with $\alpha_t^{\rm rt} > \beta_t^{{\rm rt}}$, the transaction cost with $P_i(t)$ can be alternatively written as
\begin{equation}\label{eq.Gr}
    G^{{\rm rt}}(P_i(t)) = \max\{\alpha_t^{\rm rt} P_i(t), \; \beta_t^{{\rm rt}} P_i(t)\}
\end{equation}
which is clearly convex \cite{convex}; and so is the objective in (\ref{eq.rt-prob}).

The SINR constraints in (\ref{eq.sinr}) can be actually rewritten into a convex form. Observe that an arbitrary phase rotation can be added to the beamforming vectors $\boldsymbol{w}_k(t)$ without affecting the SINRs. Hence, we can choose a phase so that $\boldsymbol{h}_{k,t}^H \boldsymbol{w}_k(t)$ is real and nonnegative. Then by proper rearrangement, the SINR constraints become convex second-order cone (SOC) constraints \cite{Wie06}; that is,
\begin{equation}\nonumber
\begin{split}
    & \sqrt{\sum_{l \neq k} |\boldsymbol{h}_{k,t}^H \boldsymbol{w}_l(t)|^2 + \sigma_k^2} \leq \frac{1}{\sqrt{\gamma_k}} \text{Re}\{\boldsymbol{h}_{k,t}^H \boldsymbol{w}_k(t)\}, \\
    &\text{Im}\{\boldsymbol{h}_{k,t}^H \boldsymbol{w}_k(t)\} =0, ~~\forall k.
\end{split}
\end{equation}

We can then rewrite the problem (\ref{eq.rt-prob}) as
\begin{align}\label{eq.rt-prob1}
        \min \; & \sum_{i\in \cal I} \biggl\{V G^{{\rm rt}}(P_c+\sum_{k\in \cal K} {\mathbf{w}_k^H(t)} \mathbf{B}_i \mathbf{w}_k(t) + P_{b,i}(t) - \frac{E_i^*[n_t]}{T})\biggr.\nonumber\\
         & ~~~~~~~~~~ \biggl.+ Q_i(n_tT) P_{b,i}(t)\}\biggr\}\nonumber\\
        \text{s. t.}\; &~ \sqrt{\sum_{l \neq k} |\boldsymbol{h}_{k,t}^H \boldsymbol{w}_l(t)|^2 + \sigma_k^2} \leq \frac{1}{\sqrt{\gamma_k}} \text{Re}\{\boldsymbol{h}_{k,t}^H \boldsymbol{w}_k(t)\}, \nonumber\\
        & ~\text{Im}\{\boldsymbol{h}_{k,t}^H \boldsymbol{w}_k(t)\} =0, ~\forall k\nonumber\\
        & ~P_{b}^{\min} \leq P_{b,i}(t) \leq P_{b}^{\max}, ~\forall i \nonumber\\
        &~P_c+\sum_{k\in \cal K} {\mathbf{w}_k^H(t)} \mathbf{B}_i \mathbf{w}_k(t) \leq P_g^{\max},~ \forall i.
\end{align}
As $G^{{\rm rt}}(\cdot)$ is convex and increasing, it is easy to see that $G^{{\rm rt}}(P_c+\sum_k {\mathbf{w}_k^H(t)} \mathbf{B}_i \mathbf{w}_k(t) + P_{b,i}(t) - E_i^*[n_t]/T)$ is jointly convex in $(P_{b,i}(t), \{\mathbf{w}_k(t)\})$ \cite[Sec.~3.2.4]{convex}. It then readily follows that (\ref{eq.rt-prob1}) is a convex optimization problem, which can be solved via off-the-shelf solvers.

\subsection{Ahead-of-Time Energy Planning}

To solve (\ref{eq.lt-prob}), the probability distribution function (pdf) of the random state $\boldsymbol{\xi}_t^{\rm rt}$ must be known across slots $t = nT, \ldots, (n+1)T-1$. However, this pdf is seldom available in practice. Suppose that $\boldsymbol{\xi}_t^{\rm rt}$ is independent and identically distributed (i.i.d.) over time slots, and takes values from a finite state space. It was proposed in \cite{Yao12} to obtain an empirical pdf of $\boldsymbol{\xi}_t^{\rm rt}$ from past realizations over a large window comprising $L$ intervals. This estimate becomes accurate as $L$ grows sufficiently large; then it can be used to evaluate the expectations in (\ref{eq.lt-prob}). Based on such an empirical pdf, an approximate solution for (\ref{eq.lt-prob}) could be obtained.

Different from \cite{Yao12}, here we propose a stochastic gradient approach to solve (\ref{eq.lt-prob}). Suppose that $\boldsymbol{\xi}_t^{\rm rt}$ is i.i.d. across time slots (but not necessarily with a finite support). For stationary $\boldsymbol{\xi}_t^{\rm rt}$, we can remove the index $t$ from all optimization variables, and rewrite (\ref{eq.lt-prob}) as (with short-hand notation $Q_i[n]:=Q_i(nT)$)
\begin{subequations}
\label{eq.lt-prob1}
\small{\begin{align}
\min \; & \sum_{i \in \cal I} \left\{V G^{{\rm lt}}(E_i[n])+T \mathbb{E}\left[VG^{{\rm rt}}(P_i(\boldsymbol{\xi}_t^{\rm rt}))+ Q_i[n]P_{b,i}(\boldsymbol{\xi}_t^{\rm rt})\right]\right\}\nonumber \\
\text{s. t.} ~
& \sqrt{\sum_{l \neq k} |\boldsymbol{h}_{k}^H \boldsymbol{w}_l(\boldsymbol{\xi}_t^{\rm rt})|^2 + \sigma_k^2} \leq \frac{1}{\sqrt{\gamma_k}} \text{Re}\{\boldsymbol{h}_{k}^H \boldsymbol{w}_k(\boldsymbol{\xi}_t^{\rm rt})\},\nonumber\\
& \text{Im}\{\boldsymbol{h}_{k}^H \boldsymbol{w}_k(\boldsymbol{\xi}_t^{\rm rt})\} =0, ~\forall k, \boldsymbol{\xi}_t^{\rm rt} \label{eq.p1a}\\
& P_{b}^{\min} \leq P_{b,i}(\boldsymbol{\xi}_t^{\rm rt}) \leq P_{b}^{\max}, ~\forall i, \boldsymbol{\xi}_t^{\rm rt} \label{eq.p1b}\\
& P_c+\sum_{k\in \cal K} {\mathbf{w}_k^H(\boldsymbol{\xi}_t^{\rm rt})} \mathbf{B}_i \mathbf{w}_k(\boldsymbol{\xi}_t^{\rm rt}) \leq P_g^{\max},~ \forall i, \boldsymbol{\xi}_t^{\rm rt} \label{eq.p1c}\\
P_c\!&+\!\!\sum_{k\in \cal K}\! {\mathbf{w}_k^H\!(\boldsymbol{\xi}_t^{\rm rt})} \mathbf{B}_i \mathbf{w}_k(\boldsymbol{\xi}_t^{\rm rt})\!+\!P_{b,i}(\boldsymbol{\xi}_t^{\rm rt}) \!=\! \frac{E_i[n]}{T} \!+\! P_i(\boldsymbol{\xi}_t^{\rm rt}), \forall i, \boldsymbol{\xi}_t^{\rm rt}\!.
\end{align}}
\end{subequations}
Note that this form explicitly indicates the dependence of the decision variables $\{P_i, P_{b,i}, \mathbf{w}_k\}$ on the realization of $\boldsymbol{\xi}_t^{\rm rt}$.

Since the energy planning problem (\ref{eq.lt-prob}) only determines the optimal ahead-of-time energy purchase $E_i^*[n]$, we can then eliminate the variable $P_i$ and write \eqref{eq.lt-prob1} as an unconstrained optimization problem with respect to the variable $E_i^*[n]$, namely
\begin{equation}\label{eq.lt-prob2}
     \min_{\{E_i[n]\}} \; \sum_{i \in \cal I} \Big[V G^{{\rm lt}}(E_i[n]) + T\bar{G}^{{\rm rt}}(\{E_i[n]\})\Big]
\end{equation}
where we define
\begin{align}\label{eq.Grt}
   \!\!\!\bar{G}^{{\rm rt}}\!(\!\{E_i[n]\}\!)\!&\!:=\!\!\!\! \min_{\{\!P_i,P_{b,i},\mathbf{w}_k\!\}} \!\sum_{i \in \cal I}\mathbb{E}\biggl\{\!V\Psi^{\rm rt}\!(E_i[n],P_{b,i}(\boldsymbol{\xi}_t^{\rm rt}),\{\mathbf{w}_k(\boldsymbol{\xi}_t^{\rm rt})\})\nonumber\\
    &\!+\!Q_i[n]P_{b,i}(\boldsymbol{\xi}_t^{\rm rt})\biggl\} ~\text{s. t. (\ref{eq.p1a}), (\ref{eq.p1b}), (\ref{eq.p1c})}
\end{align}
with the compact notation $\Psi^{{\rm rt}}(E_i, P_{b,i}, \{\mathbf{w}_k\})\!:=\! G^{{\rm rt}}(P_c\!+\!\sum_{k \in \cal K} {\mathbf{w}_k^H} \mathbf{B}_i \mathbf{w}_k\!+\! P_{b,i} \!-\! \frac{E_i}{T})$.
%
Since $\mathbb{E}[V\Psi^{{\rm rt}}(E_i[n], P_{b,i}(\boldsymbol{\xi}_t^{\rm rt}), \{\mathbf{w}_k(\boldsymbol{\xi}_t^{\rm rt})\})+ Q_i[n]P_{b,i}(\boldsymbol{\xi}_t^{\rm rt})]$ is jointly convex in $(E_i, P_{b,i}, \{\mathbf{w}_k\})$ [cf. \eqref{eq.rt-prob1}], then the minimization over $(P_{b,i}, \{\mathbf{w}_k\})$ is within a convex set; thus, \eqref{eq.p1a}-(\ref{eq.p1c}) is still convex with respect to $E_i[n]$ \cite[Sec.~3.2.5]{convex}.
In addition, due to $\alpha_n^{{\rm lt}} > \beta_n^{{\rm lt}}$, we can alternatively write $G^{{\rm lt}}(E_i[n]) = \max\{\alpha_n^{{\rm lt}}(E_i[n]-A_{i,n}), \beta_n^{{\rm lt}}(E_i[n]-A_{i,n})\}$, which is in the family of convex functions. Hence, \eqref{eq.lt-prob2} is generally a nonsmooth and unconstrained convex problem with respect to $\{E_i[n]\}$, which can be solved using the stochastic subgradient iteration described next.

The subgradient of $G^{{\rm lt}}(E_i[n])$ can be first written as
\[
    \partial G^{{\rm lt}}(E_i[n]) =
    \begin{cases}
        \alpha_n^{{\rm lt}}, ~~\text{if} ~ E_i[n]> A_{i,n} \\
        \beta_n^{{\rm lt}},  ~~\text{if} ~ E_i[n]< A_{i,n} \\
        \text{any } x \in [\beta_n^{{\rm lt}}, \alpha_n^{{\rm lt}}], ~~\text{if} ~ E_i[n]=A_{i,n}.
    \end{cases}
\]
With $\{P_{b,i}^{E}(\boldsymbol{\xi}_t^{\rm rt}), \mathbf{w}_k^{E}(\boldsymbol{\xi}_t^{\rm rt})\}$ denoting the optimal solution for the problem in (\ref{eq.Grt}), the partial subgradient of $\bar{G}^{{\rm rt}}(\{E_i[n]\})$ with respect to $E_i[n]$ is $\partial_i \bar{G}^{{\rm rt}}(\{E_i[n]\}) = V \mathbb{E}\{\partial \Psi^{{\rm rt}}(E_i[n], P_{b,i}^{E}(\boldsymbol{\xi}_t^{\rm rt}), \{\mathbf{w}_k^{E}(\boldsymbol{\xi}_t^{\rm rt})\})\}$, where
\[
    \partial \Psi^{{\rm rt}}(E_i[n], P_{b,i}^{E}(\boldsymbol{\xi}_t^{\rm rt}), \{\mathbf{w}_k^{E}(\boldsymbol{\xi}_t^{\rm rt})\}) =
    \begin{cases}
        \frac{-\beta_t^{{\rm rt}}}{T}, ~ \text{if}  ~\frac{E_i[n]}{T}> \Delta \\
        \frac{-\alpha_t^{\rm rt}}{T}, ~ \text{if} ~\frac{E_i[n]}{T} <\Delta \\
        x \in [\frac{-\alpha_t^{\rm rt}}{T}, \frac{-\beta_t^{{\rm rt}}}{T}], ~\text{else}
    \end{cases}
\]
with $\Delta:=P_c+\sum_k {{\mathbf{w}_k^{E}}^H(\boldsymbol{\xi}_t^{\rm rt})} \mathbf{B}_i \mathbf{w}_k^{E}(\boldsymbol{\xi}_t^{\rm rt})+ P_{b,i}^{E}(\boldsymbol{\xi}_t^{\rm rt})$.

Defining $\bar{g_i}(E_i):=V\partial G^{{\rm lt}}(E_i)+T\partial_i \bar{G}^{{\rm rt}}(\{E_i\})$, a standard sub-gradient descent iteration can be employed to find the optimal $E_i^*[n]$ for (\ref{eq.lt-prob2}), as
\begin{equation}\label{eq.sub}
    E_i^{(j+1)}[n] = [E_i^{(j)}[n] - \mu^{(j)} \bar{g}_i(E_i^{(j)}[n])]^+, \quad \forall i
\end{equation}
where $j$ denotes iteration index, and $\{\mu^{(j)}\}$ is the sequence of stepsizes.

Implementing (\ref{eq.sub}) essentially requires performing (high-dimensional) integration over the unknown multivariate distribution function of $\boldsymbol{\xi}_t^{\rm rt}$ present in $\bar{g}_i$ through $\bar{G}^{{\rm rt}}$ in \eqref{eq.Grt}.
To circumvent this impasse, a stochastic subgradient approach is devised based on the past realizations $\{\boldsymbol{\xi}_{\tau}^{\rm rt}, \tau =0,1,\ldots, nT-1\}$.
Per iteration $j$, we randomly draw a realization $\boldsymbol{\xi}_{\tau}^{\rm rt}$ from past realizations, and run the following iteration
\begin{equation}\label{eq.stocsub}
    E_i^{(j+1)}[n] = [E_i^{(j)}[n] - \mu^{(j)} g_i(E_i^{(j)}[n])]^+, \quad \forall i
\end{equation}
where $g_i(E_i^{(j)}[n]):=V(\partial G^{{\rm lt}}(E_i^{(j)}[n])+ T\partial \Psi^{{\rm rt}}(E_i^{(j)}[n], \linebreak P_{b,i}^{E}(\boldsymbol{\xi}_{\tau}^{\rm rt}), \{\mathbf{w}_k^{E}(\boldsymbol{\xi}_{\tau}^{\rm rt})\}))$ with
 $\{P_{b,i}^{E}(\boldsymbol{\xi}_{\tau}^{\rm rt}), \mathbf{w}_k^{E}(\boldsymbol{\xi}_{\tau}^{\rm rt})\}$ obtained by solving a convex problem \eqref{eq.Grt} with $E_i[n] =E_i^{(j)}[n]$.
As $g_i(E_i^{(j)}[n])$ is indeed an unbiased random realization of $\bar{g}_i(E_i^{(j)}[n]) = \mathbb{E} \{g_i(E_i^{(j)}[n])\}$\cite{SA}, if we adopt a sequence of non-summable diminishing stepsizes satisfying $\lim_{j \rightarrow \infty} \mu^{(j)} =0$ and $\sum_{j=0}^{\infty} \mu^{(j)} = \infty$, the iteration (\ref{eq.stocsub}) asymptotically converges to the optimal $\{E_i^*[n], \forall i\}$ as $j \rightarrow \infty$ \cite{Bertsekas09}.



Compared with \cite{Yao12}, the proposed stochastic subgradient method is particularly tailored for our setting, which does not require the random vector $\boldsymbol{\xi}_t^{\rm rt}$ to have discrete and finite support. In addition, as the former essentially belongs to the class of statistical learning based approaches \cite{Hua15}, the proposed stochastic method avoids constructing a histogram for learning the underlying multivariate distribution and requires a considerably smaller number of samples to obtain an accurate estimate of $E_i^*[n]$.

\begin{remark}
	The computational complexity of the proposed algorithm is fairly low. Specifically, for solving the real-time energy balancing and beamforming problem \eqref{eq.rt-prob1} per slot $t$, the off-the-shelf interior-point solver incurs a worst-case complexity ${\cal O}(I^{3.5} K^{3.5})$ to obtain the decisions $\{P_{b,i}^*(t), \forall i; \mathbf{w}_k^*(t), \forall k\}$ \cite{cvx}; for solving the ahead-of-time energy planning problem \eqref{eq.Grt} every $T$ slots, the stochastic subgradient approach needs ${\cal O}(1/\epsilon^2)$ iterations to obtain an $\epsilon$-optimal solution, while the per iteration complexity is in the order of ${\cal O}(I^{3.5} K^{3.5})$. And updating $E_i^{(j)}[n]$ in (27) requires only linear complexity ${\cal O}(I)$.
\end{remark}

\section{Performance Analysis}

In this section, we show that the TS-OC can yield a feasible and asymptotically (near-)optimal solution for problem (\ref{eq.prob}).

\subsection{Feasibility Guarantee}

Note that in problems (\ref{eq.lt-prob}) and (\ref{eq.rt-prob}), $\{C_i(t)\}$ are removed from the set of optimization variables and the constraints in (\ref{eq.Ci}) are ignored. While the battery dynamics $C_i(t+1) = \eta C_i(t) + P_{b,i}(t)$ are accounted for by the TS-OC algorithm (in the step of ``Queue updates''), it is not clear whether the resultant $C_i(t) \in [C^{\min}, C^{\max}]$, $\forall i, t$. Yet, we will show that by selecting a pair $(\Gamma, V)$ in (\ref{eq.GV}), we can guarantee that $C^{\min} \leq C_i(t) \leq C^{\max}$, $\forall i, t$; meaning, the online control policy produced by the TS-OC is a feasible one for the original problem (\ref{eq.prob}), under the conditions (\ref{eq.A1})--(\ref{eq.A2}).

To this end, we first show the following lemma.
\begin{lemma}
If $\bar{\alpha}:= \max\{\alpha_t^{\rm rt}, \forall t\}$ and $\underline{\beta}:= \min\{\beta_t^{{\rm rt}}, \forall t\}$, the battery (dis)charging amounts $P_{b,i}^*(t)$ obtained from the TS-OC algorithm satisfy: i) $P_{b,i}^*(t) = P_{b}^{\min}$, if $C_i(n_tT) > -V\underline{\beta} -\Gamma$; and ii) $P_{b,i}^*(t) = P_{b}^{\max}$, if $C_i(n_tT) < -V\bar{\alpha} -\Gamma$.
\end{lemma}

\begin{proof}
In TS-OC, we determine $P_{b,i}^*(t)$ by solving (\ref{eq.rt-prob}). From the equivalent problem (\ref{eq.rt-prob1}), we can see that the determination of $P_{b,i}^*(t)$ is decoupled across BSs, and it depends on the first derivative of $G^{{\rm rt}}(\cdot)$. By (\ref{eq.Gr}), the maximum possible gradient for $G^{{\rm rt}}(\cdot)$ is $V\bar{\alpha}$. It then follows that if $V\bar{\alpha}+Q_i(n_tT) <0$, we must have $P_{b,i}^*(t) = P_{b}^{\max}$. Similarly, if $V\underline{\beta}+Q_i(n_tT) >0$, we must have $P_{b,i}^*(t) = P_{b}^{\min}$. Given that $Q_i(t)=C_i(t) +\Gamma$, the lemma follows readily.
\end{proof}

Lemma 1 reveals partial characteristics of the dynamic TS-OC policy. Specifically, when the energy queue (i.e., battery SoC) is large enough,
the battery must be discharged as much as possible; that is, $P_{b,i}^*(t) = P_{b}^{\min}$. On the other hand, when the energy queue is small enough, the battery must be charged as much as possible; i.e., $P_{b,i}^*(t)  = P_{b,i}^{\max}$.
Alternatively, such results can be justified by the economic interpretation of the virtual queues. Specifically, $-\frac{Q_i(t)}{V}$ can be viewed as the instantaneous discharging price. For high prices $-\frac{Q_i(t)}{V} >\bar{\alpha}$, the TS-OC dictates full charge. Conversely, the battery units can afford full discharge if the price is low.

Based on the structure in Lemma 1, we can thus establish the following result.
\begin{proposition}
Under the conditions (\ref{eq.A1})--(\ref{eq.A2}), the TS-OC algorithm with any pair $(\Gamma, V)$ specified in (\ref{eq.GV}) guarantees $C^{\min} \leq C_i(t) \leq C^{\max}$, $\forall i$, $\forall t$.
\end{proposition}

\begin{proof}
See Appendix A.
\end{proof}

\begin{remark}
Note that Proposition 1 is a {\em sample path} result; meaning, the bounded energy queues $C_i(t) \in [C^{\min}, C^{\max}]$, $\forall i$, hold per time slot under {\em arbitrary}, even non-stationary, $\{\boldsymbol{A}_n, \alpha_n^{{\rm lt}}, \beta_n^{{\rm lt}}, \alpha_t^{\rm rt}, \beta_t^{{\rm rt}},\mathbf{H}_t\}$ processes. In other words, under the mild conditions (\ref{eq.A1})--(\ref{eq.A2}), the proposed TS-OC with proper selection of $(\Gamma, V)$ always yields a feasible control policy for (\ref{eq.prob}).
\end{remark}

\subsection{Asymptotic Optimality}

To facilitate the analysis, we assume that the random processes $\{\boldsymbol{\xi}_n^{\rm lt}\}$ and $\{\boldsymbol{\xi}_t^{\rm rt}\}$ are both i.i.d. over slow and fast timescales, respectively.
Define $\bar{C}_i := \frac{1}{NT} \sum_{t=0}^{NT-1} \mathbb{E}\{C_i(t)\}$ and $\bar{P}_{b,i} := \frac{1}{NT} \sum_{t=0}^{NT-1} \mathbb{E}\{P_{b,i}(t)\}$. Since $P_{b,i}(t) \in [P_b^{\min}, P_b^{\max}]$ and $C_i(t+1)=\eta C_i(t)+P_{b,i}(t)$, it holds that
\begin{equation}\label{relax1}
    \bar{P}_{b,i} = \frac{1}{NT} \sum_{t=0}^{NT-1}\mathbb{E}\{C_i(t+1) - \eta C_i(t)\} = (1-\eta)\bar{C}_i.
\end{equation}
As $C_i(t) \in [C^{\min}, C^{\max}]$, $\forall t$, (\ref{relax1}) then implies
\begin{equation}\label{relax2}
    (1-\eta)C^{\min} \leq \bar{P}_{b,i} \leq (1-\eta)C^{\max}, \quad \forall i.
\end{equation}

Consider now the following problem
\begin{equation}\label{eq.relax}
\begin{split}
   \tilde{\Phi}^{opt} := & \min_{{\cal X}} \;\lim_{N\rightarrow \infty} \frac{1}{NT} \sum_{t=0}^{NT-1} \sum_{i \in {\cal I}} \mathbb{E} \{\Phi_i(t)\} \\
   & \text{s. t.} ~~~ (\ref{eq.sinr}), (\ref{eq.pg}), (\ref{eq.pb}), (\ref{eq.balance}), ~\forall t, ~~ (\ref{relax2}).
\end{split}
\end{equation}
Note that the constraints in (\ref{eq.Ci}), $\forall t$, are replaced by (\ref{relax2}); i.e., the queue dynamics that need to be performed per realization per slot are replaced by a time-averaged constraint per BS~$i$. The problem (\ref{eq.relax}) is thus a relaxed version of (\ref{eq.prob}) \cite{Qin15}. Specifically, any feasible solution of (\ref{eq.prob}), satisfying (\ref{eq.Ci}), $\forall t$, also satisfies (\ref{relax2}) in (\ref{eq.relax}), due to the boundedness of $P_{b,i}(t)$ and $C_i(t)$. It then follows that $\tilde{\Phi}^{opt} \leq \Phi^{opt}$.

Variables $\{C_i(t)\}$ are removed from (\ref{eq.relax}), and other optimization variables are ``decoupled'' across time slots due to the removal of constraints (\ref{eq.Ci}). This problem has an easy-to-characterize stationary optimal control policy as formally stated in the next lemma.
\begin{lemma}
If $\boldsymbol{\xi}_n^{\rm lt}$ and $\boldsymbol{\xi}_t^{\rm rt}$ are i.i.d., there exists a stationary control policy ${\cal P}^{stat}$ that is a pure (possibly randomized) function of the current $(\boldsymbol{\xi}_{n_t}^{\rm lt}, \boldsymbol{\xi}_t^{\rm rt})$, while satisfying (\ref{eq.sinr}), (\ref{eq.pg}), (\ref{eq.pb}), (\ref{eq.balance}), and providing the following guarantees per $t$:
\begin{equation}\label{eq.opt}
\begin{split}
    & \mathbb{E}\{\sum_{i \in {\cal I}} \Phi_i^{stat}(t)\} = \tilde{\Phi}^{opt} \\
    & (1-\eta)C^{\min} \leq \mathbb{E}\{P_{b,i}^{stat}(t)\} \leq (1-\eta)C^{\max}, \;\; \forall i
\end{split}
\end{equation}
where $P_{b,i}^{stat}(t)$ denotes the decided (dis)charging amount, $\Phi_i^{stat}(t)$ the resultant transaction cost by policy ${\cal P}^{stat}$, and expectations are taken over the randomization of $(\boldsymbol{\xi}_{n_t}^{\rm lt}, \boldsymbol{\xi}_t^{\rm rt})$ and (possibly) ${\cal P}^{stat}$.
\end{lemma}

\begin{proof}
The proof argument is similar to that in \cite[Theorem 4.5]{Nee10}; hence, it is omitted for brevity.
\end{proof}

Lemma 2 in fact holds for many non-i.i.d. scenarios as well. Generalizations to other stationary processes, or even to non-stationary processes, can be found in \cite{Nee10} and \cite{Nee10-2}.

It is worth noting that (\ref{eq.opt}) not only assures that the stationary control policy ${\cal P}^{stat}$ achieves the optimal cost for (\ref{eq.relax}), but also guarantees that the resultant expected transaction cost  per slot $t$ is equal to the optimal time-averaged cost (due to the stationarity of $\boldsymbol{\xi}_{n_t}^{\rm lt}$, $\boldsymbol{\xi}_t^{\rm rt}$ and ${\cal P}^{stat}$). This plays a critical role in establishing the following result.
\begin{proposition}\label{Prop. 2}
Suppose that conditions (\ref{eq.A1})--(\ref{eq.GV}) hold. If $\boldsymbol{\xi}_n^{\rm lt}$ and $\boldsymbol{\xi}_t^{\rm rt}$ are i.i.d. across time, then the time-averaged cost under the proposed TS-OC algorithm satisfies
\[
   \lim_{N\rightarrow \infty} \frac{1}{NT} \sum_{t=0}^{NT-1} \sum_{i \in {\cal I}} \mathbb{E} \{\Phi_i^*(t)\} \leq \Phi^{opt} + \frac{M_1+M_2+M_3}{V}
\]
where the constants\footnote{Note that $\lim_{\eta\rightarrow1} \frac{1-\eta^T}{1-\eta} =T$, and $\lim_{\eta\rightarrow1} \frac{T(1-\eta)-(1-\eta^T)}{(1-\eta)(1-\eta^T)} =\frac{T-1}{2}$.}
\begin{align}
    M_1 &:= \frac{IT(1-\eta)}{2\eta(1-\eta^T)}M_B \label{eq.M1}\\
    M_2 &:= \frac{I[T(1-\eta)-(1-\eta^T)]}{(1-\eta)(1-\eta^T)} M_B \label{eq.M2} \\
    M_3 &:=  I (1-\eta) M_C \label{eq.M3}
\end{align}
with $M_B$ and $M_C$ given by
\begin{equation}\nonumber
\begin{split}
    M_B &:= \max\{[(1-\eta) \Gamma + P_b^{\min}]^2, [(1-\eta) \Gamma + P_b^{\max}]^2\} \\
    M_C &:= \max\{(\Gamma + C^{\min})^2, (\Gamma + C^{\max})^2\};
\end{split}
\end{equation}
$\Phi_i^*(t)$ denotes the resultant cost with the TS-OC, and $\Phi^{opt}$ is the optimal value of (\ref{eq.prob}) under any feasible control algorithm, including the one knowing all future realizations.
\end{proposition}

\begin{proof}
See Appendix B.
\end{proof}

\begin{remark}
Proposition 2 asserts that the proposed TS-OC algorithm ends up with a time-averaged cost having optimality gap smaller than $\frac{M_1+M_2+M_3}{V}$.
The novel TS-OC can also be viewed as a modified version of a classic queue-length based stochastic optimization scheme, where queue lengths play the role of ``stochastic'' Lagrange multipliers with a dual-subgradient solver to the regularized dual problem by subtracting an $\ell_2$-norm of Lagrange multipliers. Intuitively, the gap $M_1/V$ is inherited from the underlying stochastic subgradient method. The gap $M_2/V$ is introduced by the inaccurate queue lengths in use (since we use $Q_i(nT)$, instead of $Q_i(t)$, for all $t=nT, \ldots, (n+1)T-1$), while the gap $M_3/V$ is incurred by the presence of the $\ell_2$ regularizer in the dual function (a. k. a. the price of battery imperfections).
\end{remark}

\subsection{Main Theorem}

Based on Propositions 1 and 2, it is now possible to arrive at our main result.
\begin{theorem}\label{Them:main}
Suppose that conditions (\ref{eq.A1})--(\ref{eq.GV}) hold and $(\boldsymbol{\xi}_n^{\rm lt}, \boldsymbol{\xi}_t^{\rm rt})$ are i.i.d. over slots. Then the proposed TS-OC yields a feasible dynamic control scheme for~\eqref{eq.prob}, which is asymptotically near-optimal in the sense that
\[
    \Phi^{opt} \leq \lim_{N\rightarrow \infty} \frac{1}{NT} \sum_{t=0}^{NT-1} \sum_{i \in {\cal I}} \mathbb{E} \{\Phi_i^*(t)\} \leq \Phi^{opt} + \frac{M}{V}
\]
where $M := M_1+M_2+M_3$, as specified in Proposition \ref{Prop. 2}.
\end{theorem}

The asymptotic behavior of the proposed dynamic approach is more complicated than that of existing alternatives due to the nature of multi-scale scheduling and battery imperfections. Interesting comments on the minimum optimality gap with the TS-OC are now in order.
\begin{enumerate}
	\item [1)] When $\eta=1$ (perfect battery), the optimality gap between the TS-OC and the offline optimal scheduling reduces to
	\begin{equation*}
		\frac{M}{V}=\frac{M_1+M_2}{V}=\frac{IT}{2V}\max\{(P_b^{\min})^2, (P_b^{\max})^2\}.
	\end{equation*}
The typical tradeoff from the stochastic network optimization holds in this case \cite{Nee10}: an $\mathcal{O}(V)$ battery size is necessary, when an $\mathcal{O}(1/V)$ close-to-optimal cost is achieved. Clearly, the minimum optimality gap is given by $M/V^{\max}$, which vanishes as $V^{\max} \rightarrow \infty$. By \eqref{eq.GV3}, such an asymptotic optimality can be achieved when we have very small price difference $(\bar{\alpha}-\underline{\beta})$, or very large battery capacities $C^{\max}$.
	\item [2)] When $\eta\in(0,1)$, the constants $M_1$, $M_2$ and $M_3$ are in fact functions of $\Gamma$, whereas the minimum and maximum values of $\Gamma$ also depend on $V$ [cf. (\ref{eq.GV1})--(\ref{eq.GV2})], thus the typical tradeoff in the case 1) is no longer correct.
For a given $V^{\max}$, the minimum optimality gap, $G^{\min}(V^{\max})$, can be obtained by solving the following problem:
\begin{equation}\label{eq.gap}
     \min_{(V, \Gamma)} ~\frac{M}{V}\!=\!\frac{M_1(\Gamma)}{V}\!+\!\frac{M_2(\Gamma)}{V}\!+\!\frac{M_3(\Gamma)}{V}, ~\text{s. t. } ~(\ref{eq.GV}).
\end{equation}
For $V \geq 0$, we know that the quadratic-over-linear functions $\frac{[(1-\eta) \Gamma + P_b^{\min}]^2}{V}$ and $\frac{[(1-\eta) \Gamma + P_b^{\max}]^2}{V}$ are jointly convex in $V$ and $\Gamma$ \cite{convex}. As a point-wise maximum of these two convex functions, $\frac{M_B(\Gamma)}{V}$ is also convex \cite{convex}. Then $\frac{M_1(\Gamma)}{V}$ and $\frac{M_2(\Gamma)}{V}$ are clearly convex by \eqref{eq.M1}-\eqref{eq.M2}; and likewise for $\frac{M_3(\Gamma)}{V}$. Since the objective is convex and the constraints in (\ref{eq.GV}) are linear, problem (\ref{eq.gap}) is a convex program which can be efficiently solved by general interior-point methods. Note that $G^{\min}(V^{\max})$ no longer monotonically decreases with respect to $V^{\max}$ (or $C^{\max}$); 
see also \cite{Qin15}. This makes sense intuitively because for a large battery capacity, the impact of using inaccurate queue lengths (battery SoC) and the dissipation loss due to battery imperfections will also be enlarged. 
The smallest possible optimality gap can be numerically computed by one dimensional search over $G^{\min}(V^{\max})$ with respect to $V^{\max}$.

\end{enumerate}

\section{Numerical Tests}\label{sec:test}

In this section, simulated tests are presented to evaluate our proposed TS-OC algorithm, and justify the analytical claims in Section~IV.

\subsection{Experiment Setup}
The considered CoMP network includes $I=2$ BSs each with $M=2$ transmit antennas, and $K=3$ mobile users. The system bandwidth is 1~MHz, and each element in channel vectors $\boldsymbol{h}_{ik,t}, \forall i,k,t$, is a zero-mean complex-Gaussian random variable with unit variance. Each coarse-grained interval consists of $T=5$ time slots. The limits of $P_{g,i}$, $P_{b,i}$ and $C_i$, as well as the values of the initial SoC $C_i(0)$ and $P_c$ are listed in Table~I. The battery storage efficiency is $\eta=0.95$.
The ahead-of-time and real-time energy purchase prices $\alpha_n^{{\rm lt}}$ and $\alpha_t^{\rm rt}$ are generated from folded normal distributions, with $\mathbb{E}\{\alpha_n^{{\rm lt}}\}=1.15$ and $\mathbb{E}\{\alpha_t^{\rm rt}\}=2.3$. The selling prices are set as $\beta_n^{{\rm lt}}=0.9\times \alpha_n^{{\rm lt}}$ and $\beta_t^{{\rm rt}}=0.3 \times\alpha_t^{\rm rt}$.
The harvested energy $A_{i,n}$ is also generated from a folded normal distribution. Finally, the Lyapunov control parameter $V$ is chosen as $V=V^{\max}$.
The proposed TS-OC algorithm is compared with three baseline schemes to benchmark its performance.
ALG~1 is a one-scale scheme without ahead-of-time energy planning;
ALG~2 performs two-scale online control without leveraging the renewable energy or energy storage devices; and the offline benchmark is an ideal scheme with a-priori knowledge of  future channel states, energy prices and RES arrival realizations.
\begin{table}[t]\addtolength{\tabcolsep}{1pt}
\centering
\caption{Parameter Values. All units are kWh.}\label{tab:param}
    \begin{tabular}{ c|c|c|c|c|c|c }
    \hline
$P_c$ &$P_{g}^{\max}$ &$P_{b}^{\min}$   &$P_{b}^{\max}$ &$C^{\min}$ &$C^{\max}$  &$C_i(0)$ \\ \hline
10    & 50        & -2           & 2  	  & 0    & 80         &0 \\
\hline
    \end{tabular}
    \vspace{-0.2cm}
\end{table}

\subsection{Numerical Results}

Fig.~\ref{fig: Avecost} shows the running-average transaction costs of the proposed algorithm, ALGs~1-2, as well as the offline benchmark. It is seen that within 500 time slots, the proposed approach converges the closest to the lower bound, while ALGs 1-2 incur about 71\% and 31\% larger costs than the proposed one.
However, note that the optimal offline counterpart cannot work in practice due to the lack of future. In addition, the optimality gap can be reduced as the battery efficiency $\eta$ approaches 1.
Among online schemes, the TS-OC algorithm intelligently takes advantage of the ahead-of-time energy planning, and the renewable energy and batteries, to hedge against future potential high energy cost, while ALGs~1-2 have to purchase much more expensive energy from the real-time energy market and result in a higher transaction cost.

The theoretical optimality-gap [cf. \eqref{eq.gap}] between the TS-OC and the offline optimal scheduling is depicted in Fig.~\ref{fig: Gap} under different battery capacities $C^{\max}$.
As analyzed after Theorem~1, the optimality-gap $M/V$ for $\eta=1$ diminishes as $C^{\max}$ (or $V^{\max}$) grows; whereas the gaps for $\eta=0.9$ and $\eta=0.95$ are no longer monotonically decreasing. Specifically, both of them first decrease and then increase, reaching the lowest points (where the optimality gaps are minimized) at $C^{\max}=40$~kWh and $C^{\max}=55$~kWh, respectively.
As expected, the gap for the worst storage efficiency $\eta=0.9$ remains the largest across the entire spectrum of battery capacity. 

In Fig.~\ref{fig: Avecosteta}, the average transaction cost of the TS-OC is compared under different battery efficiencies $\eta=0.9, 0.95, 1$. Clearly, the average costs monotonically decrease as $C^{\max}$ grows. The BSs with imperfect batteries ($\eta=0.9, 0.95$) require larger budgets for energy purchase than the ones with perfect batteries ($\eta=1$), thus compensating for the battery degeneration losses. In particular, when $C^{\max}=120$~kWh, the costs for $\eta=0.9$ and $\eta=0.95$ are 41.8\% and 33.8\% larger than that of the perfect battery case, respectively. 

The evolutions of battery SoC $C_1(t)$ with different storage efficiencies $\eta$ are compared in Fig.~\ref{fig: Ci}.
Clearly, all the three lines fluctuate within the feasible region; i.e., $C^{\min}\leq C_1(t) \leq C^{\max}$.
Among the three cases, the battery with $\eta=1$ maintains the highest energy level, followed by those with $\eta=0.95$ and $\eta=0.9$.
Intuitively speaking, keeping a high energy level in an imperfect battery results in much higher energy dissipation losses.
As a  result, the TS-OC algorithm tends to maintain a low energy level in such cases (e.g., around 30 kWh for $\eta=0.9$) to reduce average energy loss, and (dis)charge the battery less frequently.

\begin{figure}[t]
\centering
\includegraphics[width=0.5\textwidth]{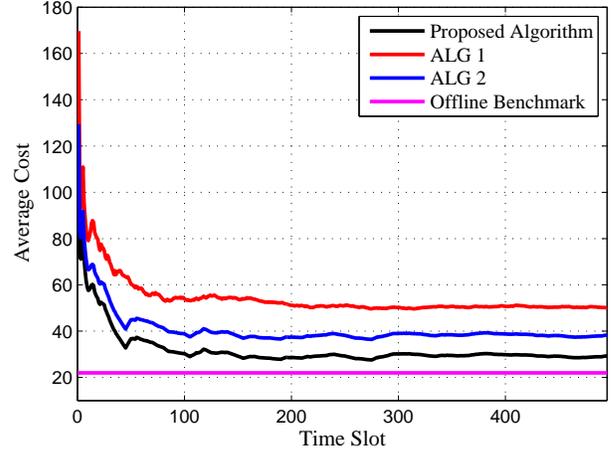}
\vspace{-0.5cm}
\caption{Comparison of average transaction cost.}
\label{fig: Avecost}
\vspace{-0.2cm}
\end{figure}

\begin{figure}[t]
\centering
\includegraphics[width=0.5\textwidth]{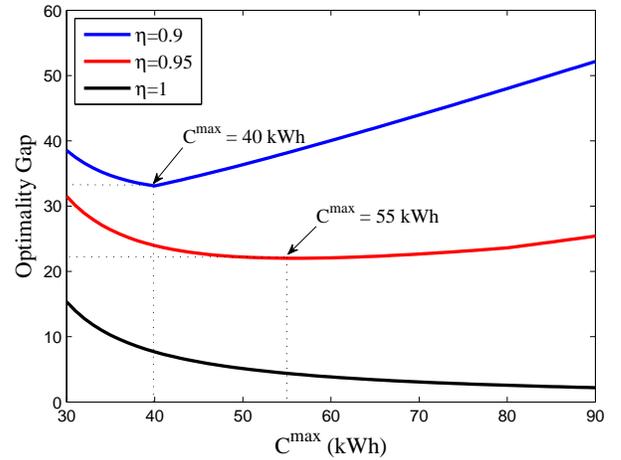}
\vspace{-0.5cm}
\caption{Optimality-gap versus battery capacity $C^{\max}$.}
\label{fig: Gap}
\vspace{-0.2cm}
\end{figure}

\begin{figure}[t]
\centering
\includegraphics[width=0.5\textwidth]{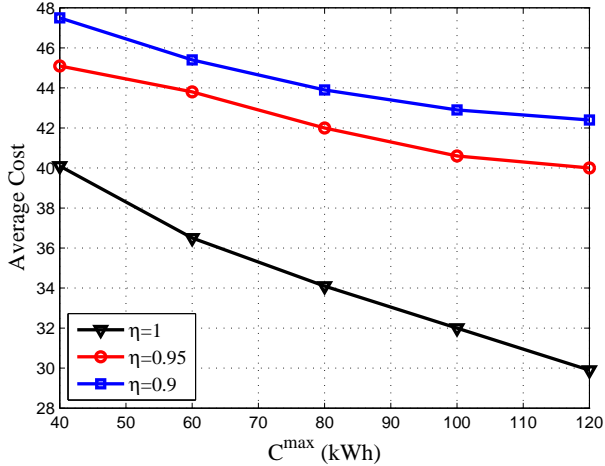}
\vspace{-0.5cm}
\caption{Average transaction cost versus battery capacity $C^{\max}$.}
\label{fig: Avecosteta}
\vspace{-0.2cm}
\end{figure}

\begin{figure}[t]
\centering
\includegraphics[width=0.5\textwidth]{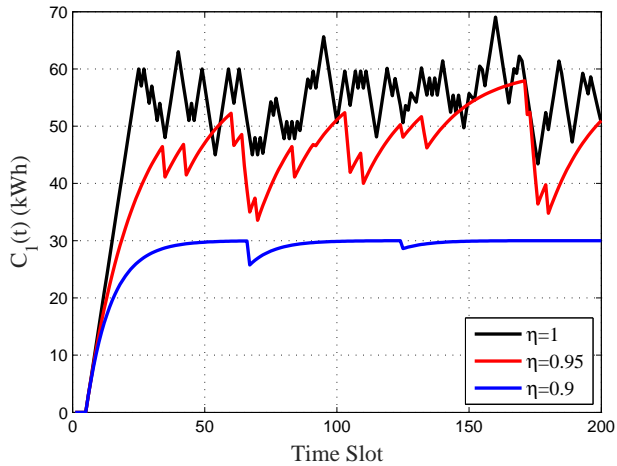}
\vspace{-0.5cm}
\caption{TS-OC based schedule of the battery SoC $C_1(t)$.}
\label{fig: Ci}
\vspace{-0.2cm}
\end{figure}

The previous remarks are further substantiated by Fig.~\ref{fig: optLag}, where the instantaneous discharging price, or, the ``stochastic'' Lagrange multiplier $-\frac{Q_1(t)}{V}$ is compared with the running-average purchase and selling prices $\bar{\alpha}_t^{{\rm rt}}:=(1/t)\sum_{\tau=1}^t \alpha_{\tau}^{{\rm rt}}$ and $\bar{\beta}_t^{{\rm rt}}:=(1/t)\sum_{\tau=1}^t \beta_{\tau}^{{\rm rt}}$.
It is interesting to observe that with a perfect battery ($\eta=1$), the instantaneous discharging price $-\frac{Q_1(t)}{V}$ is hovering between the average purchase and selling prices, which features a frequent (dis)charging operation. For $\eta=0.95$ or $\eta=0.9$, $-\frac{Q_1(t)}{V}$ is relatively high compared to the average purchase and selling prices, which discourages frequent (dis)charging; see also Fig. \ref{fig: Ci}. Note that the evolution of $-\frac{Q_1(t)}{V}$ can be further linked to the standard results from sensitivity analysis, which implies that the subdifferential of the objective $\lim_{N\rightarrow \infty} \frac{1}{NT} \sum_{t=0}^{NT-1} \sum_i \mathbb{E} \{\Phi_i(t)\}$ with respect to $P_{b,i}(t)$ (the convex hull of average purchase and selling prices) coincides with the negative of the optimal dual variable corresponding to \eqref{relax1} \cite{convex}.
Building upon this claim, the asymptotic optimality can be easily verified for $\eta=1$ since the ``stochastic'' Lagrange multiplier $-\frac{Q_1(t)}{V}$ converges to a neighborhood of the optimal dual variable; and a large optimality gap is also as expected for the imperfect batteries $\eta<1$ due to the distance between $-\frac{Q_1(t)}{V}$ and the average purchase and selling prices.

\begin{figure}[t]
\centering
\includegraphics[width=0.5\textwidth]{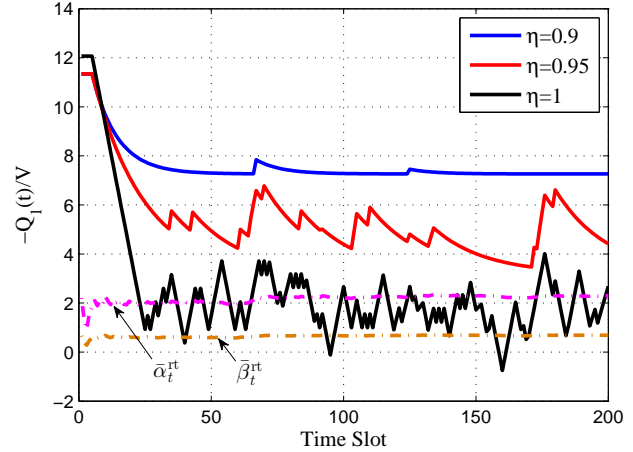}
\vspace{-0.5cm}
\caption{The evolution of $-Q_1(t)/V$ and running-average of energy prices.}
\label{fig: optLag}
\vspace{-0.2cm}
\end{figure}

\begin{figure}[t]
\centering
\includegraphics[width=0.5\textwidth]{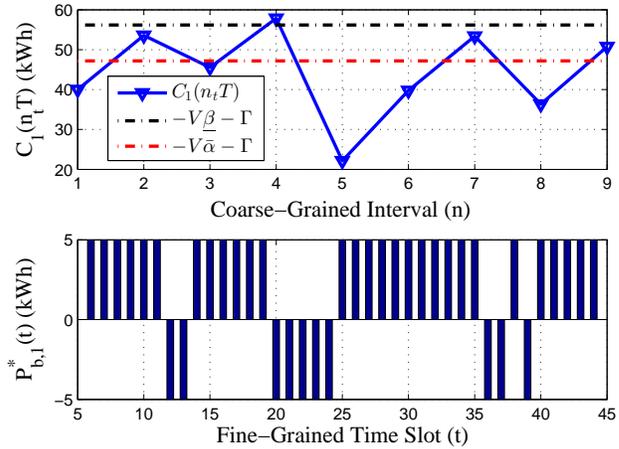}
\vspace{-0.5cm}
\caption{TS-OC based schedule of the battery SoC $C_1(n_tT)$ and battery (dis)charging actions $P_{b,1}^*(t)$, where $P_b^{\max} = 5$~kWh and $P_b^{\min} = -5$~kWh.}
\label{fig: Pb}
\vspace{-0.2cm}
\end{figure}

Taking a deeper look, the battery SoC $C_1(n_tT)$ and the real-time battery (dis)charging $P^*_{b,1}(t)$ are jointly depicted in Fig.~\ref{fig: Pb} to reveal the (dis)charging characteristics stated in Lemma~1. It can be observed that the TS-OC dictates the full discharge $P^*_{b,1}(t)=P_b^{\min}$ in the incoming 5 fine-grained slots $t \in [20,24]$ when $C_1(n_tT) > -V\underline{\beta} -\Gamma$ at $n = 4$, while the battery is fully charged $P^*_{b,1}(t)=P_b^{\max}$ when $C_1(n_tT) < -V\bar{\alpha} -\Gamma$ at $n = 1, 3, 5, 6, 8$.
In addition, when $C_1(n_tT)\in [-V\bar{\alpha} -\Gamma, -V\underline{\beta} -\Gamma]$ at $n = 2, 7$, $P^*_{b,1}(t)$ must be obtained by solving \eqref{eq.rt-prob1} numerically. 

\begin{figure}[t]
\centering
\includegraphics[width=0.5\textwidth]{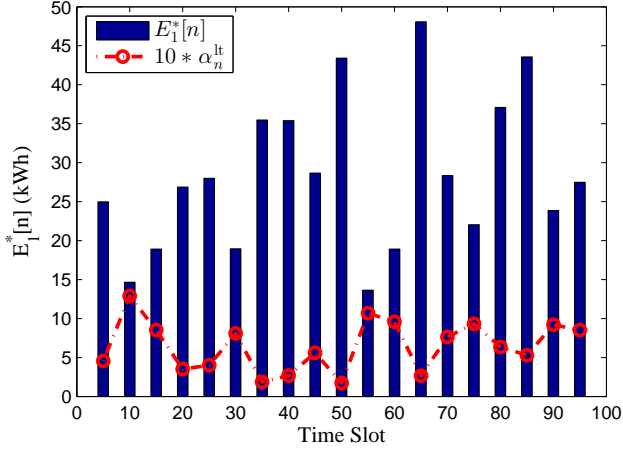}
\vspace{-0.5cm}
\caption{TS-OC based schedule of the optimal energy planning $E^*_1[n]$.}
\label{fig: Ei}
\vspace{-0.2cm}
\end{figure}

Fig.~\ref{fig: Ei} shows the optimal energy planning $E^*_1[n]$ over a 100-slot period, along with the fluctuating ahead-of-time energy purchase prices $\alpha_n^{{\rm lt}}$ for the resultant online policy.
One observation is that the ahead-of-time energy purchase $E^*_1[n]$ highly depends on the long-term price $\alpha_n^{{\rm lt}}$.
Specifically, the proposed scheme tends to request more energy for future $T$ slots when $\alpha_n^{{\rm lt}}$ is lower (e.g., $n=10,13,17$), and tends to purchase less energy when $\alpha_n^{{\rm lt}}$ is higher (e.g., $n=2,11$).

%

\section{Conclusions}

A two-scale dynamic resource allocation task was considered for RES-integrated CoMP transmissions. Taking into account the variability of channels, RES and ahead-of-time/real-time electricity prices, as well as battery imperfections, a stochastic optimization problem was formulated to minimize the long-term average energy transaction cost subject to the QoS requirements. Capitalizing on the Lyapunov optimization technique and the stochastic subgradient iteration, a two-scale online algorithm was developed to make control decisions `on-the-fly.' It was analytically established that the novel approach yields feasible and asymptotically near-optimal resource schedules without knowing any statistics of the underlying stochastic processes. Simulated tests confirmed the merits of the proposed approach and highlighted the effect of battery imperfections on the proposed online scheme.
This novel two-scale optimization framework opens up some interesting research directions, which include incorporating the power network constraints and/or transmission losses in the formulation, pursuing a fast convergent approach by learning from historical system statistics, and reducing the battery size leveraging the so-called predictive scheduling.

\appendix
\subsection{Proof of Proposition 1}

The proof proceeds by induction. First, set $C_i(0) \in [C^{\min}, C^{\max}]$, $\forall i$, and suppose that this holds for all $C_i(nT)$ at slot $nT$. We will show the bounds hold for $C_i(t)$, $\forall t =nT+1, \ldots, (n+1)T$, as well as in subsequent instances.

By $C_i(t+1)=\eta C_i(t) +P_{b,i}^*(t)$, we have
\begin{align}\label{eq.Cit}
    C_i(t) = \eta^{t-nT} C_i(nT) + \sum_{\tau=nT}^{t-1} [\eta^{t-1-\tau} P_{b,i}^*(\tau)],~~~ \nonumber\\
     \forall t =nT+1, \ldots, (n+1)T.
\end{align}

Note that by the definitions of $\Gamma^{\min}$ and $\Gamma^{\max}$ in (\ref{eq.GV1})-(\ref{eq.GV2}), we have $C^{\min} \leq -V \bar{\alpha} -\Gamma < -V \underline{\beta} -\Gamma \leq C^{\max}$. We then consider the following three cases.
\begin{itemize}
\item[c1)] If $C_i(nT) \in [C^{\min}, -V \bar{\alpha} -\Gamma)$, then Lemma~1 implies that $P_{b,i}^*(t) = P_{b}^{\max}$, $\forall t = nT, \ldots, (n+1)T-1$. From (\ref{eq.Cit}), we have, $\forall t=nT+1, \ldots, (n+1)T$,
    \begin{itemize}
    \item[i)] $C_i(t) \geq \eta^{t-nT} C^{\min} + \frac{1-\eta^{t-nT}}{1-\eta} P_{b}^{\max} \geq C^{\min}$, due to the condition (\ref{eq.A1});
    \item[ii)] $C_i(t) \leq \eta^{t-nT} (-V \bar{\alpha} -\Gamma) + \frac{1-\eta^{t-nT}}{1-\eta} P_{b}^{\max} \leq \eta^{t-nT}(-V\underline{\beta} -\Gamma) + \frac{1-\eta^{t-nT}}{1-\eta} P_{b}^{\max} \leq C^{\max}$, due to $\underline{\beta} < \bar{\alpha}$, $\Gamma \geq \Gamma^{\min}$, and the definition of $\Gamma^{\min}$ in (\ref{eq.GV1}).
    \end{itemize}

\item[c2)] If $C_i(nT) \in [-V \bar{\alpha} -\Gamma, -V \underline{\beta} -\Gamma]$, then $P_{b,i}^*(t) \in [P_b^{\min}, P_{b}^{\max}]$. We have, $\forall t=nT+1, \ldots, (n+1)T$,
    \begin{itemize}
    \item[i)] $C_i(t) \geq \eta^{t-nT} (-V \bar{\alpha} -\Gamma) + \frac{1-\eta^{t-nT}}{1-\eta} P_{b}^{\min} \geq C^{\min}$, due to $\Gamma \leq \Gamma^{\max}$ and the definition of $\Gamma^{\max}$ in (\ref{eq.GV2});
    \item[ii)] $C_i(t) \leq \eta^{t-nT} (-V \underline{\beta} -\Gamma) + \frac{1-\eta^{t-nT}}{1-\eta} P_{b}^{\max} \leq C^{\max}$, as with c1-ii); and
    \end{itemize}

\item[c3)] If $C_i(nT) \in (-V \underline{\beta} -\Gamma, C^{\max}]$, it follows from Lemma~1 that $P_{b,i}^*(t) = P_{b}^{\min}$, $\forall t = nT, \ldots, (n+1)T-1$. We have, $\forall t=nT+1, \ldots, (n+1)T$
    \begin{itemize}
    \item[i)]  $C_i(t) \geq \eta^{t-nT} (-V \underline{\beta} -\Gamma) + \frac{1-\eta^{t-nT}}{1-\eta} P_{b}^{\min} \geq \eta^{t-nT} (-V \bar{\alpha} -\Gamma) + \frac{1-\eta^{t-nT}}{1-\eta} P_{b}^{\min} \geq C^{\min}$, due to $\underline{\beta} < \bar{\alpha}$ and c2-i);
    \item[ii)] $C_i(t) \leq \eta^{t-nT} C^{\max} + \frac{1-\eta^{t-nT}}{1-\eta} P_{b}^{\min} \leq C^{\max}$, due to $\eta \leq 1$, and $P_{b}^{\min} < 0$.
    \end{itemize}
\end{itemize}

Cases c1)--c3) together prove the proposition.

\subsection{Proof of Proposition 2}
The evolution of $Q_i(t)$ in the TS-OC is given by $Q_i(t+1)=C_i(t+1)+\Gamma = \eta C_i(t) + P_{b,i}^*(t) +\Gamma = \eta Q_i(t) + (1-\eta) \Gamma + P_{b,i}^*(t)$. Hence, we have
\begin{align}
    & [Q_i(t+1)]^2 =[\eta Q_i(t) + (1-\eta) \Gamma + P_{b,i}^*(t)]^2 \notag \\
    & ~= \eta^2 [Q_i(t)]^2 + 2 \eta Q_i(t)[(1-\eta) \Gamma + P_{b,i}^*(t)] \notag \\
    & ~~~~~ + [(1-\eta) \Gamma + P_{b,i}^*(t)]^2 \notag \\
    & ~\leq \eta^2 [Q_i(t)]^2 + 2 \eta Q_i(t)[(1-\eta) \Gamma + P_{b,i}^*(t)] \notag \\
    & ~~~~~ + \max\{[(1-\eta) \Gamma + P_b^{\min}]^2, [(1-\eta) \Gamma + P_b^{\max}]^2\} \notag
\end{align}
where the last inequality holds due to \eqref{eq.pb}.

With $\mathbf{Q}(t): = [Q_1(t),\ldots, Q_I(t)]'$, consider the Lyapunov function $L(\mathbf{Q}(t)):=\frac{1}{2} \sum_i [Q_i(t)]^2$. Using the short-hand notation $\boldsymbol{Q}[n]:= \boldsymbol{Q}(nT)$, it readily follows that
\begin{align}
    & \triangle_T (\mathbf{Q}[n]):=L(\mathbf{Q}[n+1])-L(\mathbf{Q}[n]) \notag \\
    & ~~ \leq -\frac{1}{2}(1-\eta^2)\sum_{t=nT}^{(n+1)T-1} \sum_{i\in \cal I} [Q_i(t)]^2+\frac{IT}{2} M_B \notag \\
    & ~~~~~~~~+ \sum_{t=nT}^{(n+1)T-1} \sum_{i\in \cal I} \{\eta Q_i(t)[(1-\eta) \Gamma + P_{b,i}^*(t)] \} \notag\\
    & ~~ \leq \frac{IT}{2} M_B + \sum_{t=nT}^{(n+1)T-1} \sum_{i\in \cal I}  \{\eta Q_i(t)[(1-\eta) \Gamma + P_{b,i}^*(t)] \}. \notag
\end{align}
Since $Q_i(t+1)=\eta Q_i(t) + (1-\eta) \Gamma + P_{b,i}^*(t)$ and $P_b^{\min} \leq P_{b,i}^*(t) \leq P_b^{\max}$, we have: $\forall t = nT, \ldots, (n+1)T-1$,
\begin{align*}
    \eta^{t-nT} Q_i[n] + \frac{1-\eta^{t-nT}}{1-\eta}[(1-\eta)\Gamma + P_b^{\min}] \leq Q_i(t) \\
    \leq \eta^{t-nT} Q_i[n] + \frac{1-\eta^{t-nT}}{1-\eta}[(1-\eta)\Gamma + P_b^{\max}].
    \end{align*}
This implies that $\forall t = nT, \ldots, (n+1)T-1$,
\begin{align*}
    & Q_i(t)[(1-\eta) \Gamma + P_{b,i}^*(t)] \leq \eta^{t-nT}Q_i[n][(1-\eta) \Gamma + P_{b,i}^*(t)]\\
    & + \frac{1-\eta^{t-nT}}{1-\eta} \max\{[(1-\eta) \Gamma + P_b^{\min}]^2, [(1-\eta) \Gamma + P_b^{\max}]^2\}.
\end{align*}
Consequently, it follows that
\begin{align*}
    & \triangle_T (\mathbf{Q}[n]) \leq  \frac{IT}{2} M_B + \sum_{t=nT}^{(n+1)T-1} \sum_{i\in \cal I} \{\frac{\eta(1-\eta^{t-nT})}{1-\eta} M_B \\
    & ~~~~+ \sum_{t=nT}^{(n+1)T-1} \sum_{i\in \cal I} \{\eta^{t-nT+1} Q_i[n][(1-\eta) \Gamma + P_{b,i}^*(t)] \} \\
    & ~\leq  \frac{IT}{2} M_B + \frac{I \eta[T(1-\eta)-(1-\eta^T)]}{(1-\eta)^2} M_B \\
    & ~~~~+ \sum_{t=nT}^{(n+1)T-1} \sum_{i\in \cal I} \{\eta^{t-nT+1} Q_i[n][(1-\eta) \Gamma + P_{b,i}^*(t)] \}.
\end{align*}
Taking expectations and adding $\sum_{t=nT}^{(n+1)T-1} \sum_i [\eta^{t-nT+1} \linebreak V \mathbb{E}\{\Phi_i^*(t)\}]$ to both sides, we arrive at (with short-hand notation $M_{\triangle}:=\frac{IT}{2} M_B + \frac{I\eta[T(1-\eta)-(1-\eta^T)]}{(1-\eta)^2} M_B$):
\begin{align}
    & \mathbb{E}\{\triangle_T (\mathbf{Q}[n])\} + \sum_{t=nT}^{(n+1)T-1} \eta^{t-nT+1} \sum_{i\in \cal I} [V \mathbb{E}\{\Phi_i^*(t)\}] \notag \\
    & \leq M_{\triangle}+ \sum_{t=nT}^{(n+1)T-1} [\eta^{t-nT+1} \sum_{i\in \cal I} Q_i[n](1-\eta) \Gamma] \notag \\
    & ~~~~+ \sum_{t=nT}^{(n+1)T-1}[ \eta^{t-nT+1}  \sum_{i\in \cal I} \mathbb{E}\{V\Phi_i^*(t)+Q_i[n]P_{b,i}^*(t)\}] \notag \\
    & = M_{\triangle}+ \sum_{t=nT}^{(n+1)T-1} [\eta^{t-nT+1} \sum_{i\in \cal I} Q_i[n](1-\eta) \Gamma] \notag \\
    & ~~~~+ \frac{\eta(1-\eta^T)}{(1-\eta)T} \sum_{t=nT}^{(n+1)T-1}[ \sum_{i\in \cal I} \mathbb{E}\{V\Phi_i^*(t)+Q_i[n]P_{b,i}^*(t)\}] \notag \\
    & \leq M_{\triangle}+ \sum_{t=nT}^{(n+1)T-1} [\eta^{t-nT+1} \sum_{i\in \cal I} Q_i[n](1-\eta) \Gamma] \notag \\
    & ~~~~+ \frac{\eta(1-\eta^T)}{(1-\eta)T} \sum_{t=nT}^{(n+1)T-1}[ \sum_{i\in \cal I} \mathbb{E}\{V\Phi_i^{stat}(t)+Q_i[n]P_{b,i}^{stat}(t)\}] \notag \\
    & = M_{\triangle} + \sum_{t=nT}^{(n+1)T-1} [ \eta^{t-nT+1} \sum_{i\in \cal I} \mathbb{E} \{V\Phi_i^{stat}(t)\} ]\notag \\
    & ~~~~+ \sum_{t=nT}^{(n+1)T-1} [\eta^{t-nT+1}Q_i[n]\{(1-\eta) \Gamma + P_{b,i}^{stat}(t)\}] \notag \\
    & \leq M_{\triangle} + I\eta (1-\eta^T) M_C+ \frac{\eta(1-\eta^T)}{1-\eta} V\tilde{\Phi}^{opt} \notag
\end{align}
where the two equalities hold since both $\sum_{i\in \cal I} \mathbb{E}\{V\Phi_i^*(t)+Q_i[n]P_{b,i}^*(t)\}$ for the TS-OC and $\sum_{i\in \cal I} \mathbb{E}\{V\Phi_i^{stat}(t)+Q_i[n]P_{b,i}^{stat}(t)\}$ for ${\cal P}^{stat}$ are in fact the same for slots $t=nT, \ldots, (n+1)T-1$, when $\boldsymbol{\xi}_t^{\rm rt}$ is i.i.d. over slots; the second inequality is because the TS-OC algorithm minimizes the third term $\sum_i \mathbb{E}\{V\Phi_i(t)+Q_i[n]P_{b,i}(t)\}]$ among all policies satisfying (\ref{eq.sinr}), (\ref{eq.pg}), (\ref{eq.pb}), and (\ref{eq.balance}), including ${\cal P}^{stat}$; and the last inequality is due to (\ref{eq.opt}) and $Q_i[n] \in [C^{\min}+\Gamma, C^{\max} + \Gamma]$ under conditions (\ref{eq.A1})--(\ref{eq.GV}) per Proposition 1.

Again, note that $\sum_i [V \mathbb{E}\{\Phi_i^*(t)\}]$ for the TS-OC is the same for slots $t=nT, \ldots, (n+1)T-1$, when $\boldsymbol{\xi}_t^{\rm rt}$ is i.i.d. over slots. Summing over all $n=1,2,\ldots$, we then have
\begin{align}
    & \sum_{n=0}^{N-1}\mathbb{E}\{\triangle_T (\mathbf{Q}[n])\} + \sum_{n=0}^{N-1} \sum_{t=nT}^{(n+1)T-1} \eta^{t-nT+1} \sum_{i\in \cal I} [V \mathbb{E}\{\Phi_i^*(t)\}] \notag \\
    & = \mathbb{E}[L(\mathbf{Q}[N])]-L(\mathbf{Q}[0]) + \frac{\eta(1-\eta^T)}{(1-\eta)T} \sum_{t=0}^{NT-1} \sum_{i\in \cal I} [V \mathbb{E}\{\Phi_i^*(t)\}] \notag \\
    & \leq N[M_{\triangle} + I\eta (1-\eta^T) M_C+ \frac{\eta(1-\eta^T)}{1-\eta} V\tilde{\Phi}^{opt}] \notag
\end{align}
which leads to
\begin{align}
    &\frac{1}{NT} \sum_{t=0}^{NT-1} \mathbb{E}[\sum_{i\in \cal I} \mathbb{E}\{\Phi_i^*(t)\}] \notag \\
    & ~~~ \leq \tilde{\Phi}^{opt} + \frac{M_1+M_2+M_3}{V} + \frac{(1-\eta)}{\eta(1-\eta^T)}\frac{L(\mathbf{Q}[0])}{NV} \notag \\
    &  ~~~ \leq \Phi^{opt} + \frac{M_1+M_2+M_3}{V} + \frac{(1-\eta)}{\eta(1-\eta^T)}\frac{L(\mathbf{Q}[0])}{NV} \notag
\end{align}
and the proposition follows by taking the limit as $N \rightarrow \infty$.

\end{document}